\documentclass[a4paper]{article}

\usepackage{microtype}

\usepackage{amstext}
\usepackage{amsmath}
\makeatletter
\newcommand{\noun}[1]{\textsc{#1}}
\usepackage{tikz}
\usetikzlibrary{arrows.meta}
\usetikzlibrary{backgrounds}

\usepackage{amsthm}
\usepackage{amsfonts}
\newtheorem{theorem}{Theorem}
\newtheorem{definition}[theorem]{Definition}
\newtheorem{corollary}[theorem]{Corollary}
\newtheorem{lemma}[theorem]{Lemma}
\usepackage[margin=3cm]{geometry}

\usepackage{thm-restate}

\makeatother

\newtheorem{observation}[theorem]{Observation}

\usepackage{xspace}
\newcommand{\opt}[1]{\ensuremath{S^{\star}_{#1}}\xspace}
\newcommand{\vopt}[1]{\ensuremath{f^{\star}_{#1}\xspace}}
\newcommand{\dopt}[1]{\ensuremath{\delta^{\star}_{#1}\xspace}}

\newcommand{\inc}[1]{\ensuremath{\vec{S}_{#1}\xspace}}
\newcommand{\vinc}[1]{\ensuremath{f(\inc{#1})\xspace}}

\newcommand{\set}[1]{ \left\{ #1 \right\} }

\let\phi\varphi

\newcommand{\far}{\alpha}
\newcommand{\ratio}{q}
\newcommand{\keytime}{t}
\newcommand{\function}{h}
\newcommand{\eps}{\varepsilon}

\newcommand{\N}{\mathbb{N}}

\newcommand{\eq}{\Longleftrightarrow}
\newcommand{\ama}{\[ \begin{aligned}}
\newcommand{\ema}{\end{aligned} \]}
\newcommand{\Greedy}{greedy\xspace}
\newcommand{\IMF}{\noun{Maximum Bridge-Flow}\xspace}
\newcommand{\setc}[2]{\left\{#1\ \right|\ \left.#2\right\}}
\newcommand{\Wlog}{{Without loss of generality}}

\newcommand{\noregions}{\ensuremath{N}\xspace}
\newcommand{\regionsin}{\ensuremath{\set{1,\dots,\noregions}}\xspace}
\newcommand{\densityr}[1]{\delta(#1)}
\newcommand{\vallb}[1]{v(#1)}
\newcommand{\maxpicks}{m}

\newcommand{\val}{\ensuremath{\text{\sc val}}\xspace}
\newcommand{\flow}{f}
\newcommand{\sm}{\setminus}

\newcommand{\fnbr}{f^{nb}_r}
\newcommand{\profit}[1]{p_{#1}} 

\definecolor{orange}{RGB}{255,127,0}
\newcommand{\boldall}[1]{\ifmmode\mathbf{#1}\else\textbf{\boldmath{#1}}\fi}

\begin{document}

\title{General Bounds for Incremental Maximization}

\author{Aaron Bernstein$^1$, Yann Disser$^2$\thanks{Supported by the `Excellence Initiative' of the German Federal and State Governments and the Graduate School~CE at TU~Darmstadt.}\,, and Martin Gro\ss$^3$\thanks{Supported by the German Research Foundation (DFG) within project A07 of CRC TRR 154.}\\
\small{$^1$ TU Berlin, Germany, \texttt{bernstei@gmail.com}}\\
\small{$^2$ TU Darmstadt, Germany, \texttt{disser@mathematik.tu-darmstadt.de}}\\
\small{$^3$ University of Waterloo, Canada, \texttt{mgrob@uwaterloo.ca}}}





\newcommand{\solelements}{U\xspace} 
\newcommand{\weaksub}{augmentable\xspace} 

\maketitle

\begin{abstract}
  We propose a theoretical framework to capture incremental solutions to cardinality constrained maximization problems.
  The defining characteristic of our framework is that the cardinality/support of the solution is bounded by a value~$k\in\mathbb{N}$ that grows over time, and we allow the solution to be extended one element at a time.
  We investigate the best-possible competitive ratio of such an incremental solution, i.e., the worst ratio over all~$k$ between the incremental solution after~$k$ steps and an optimum solution of cardinality~$k$.
  We define a large class of problems that contains many important cardinality constrained maximization problems like maximum matching, knapsack, and packing/covering problems.
  We provide a general~$2.618$-competitive incremental algorithm for this class of problems, and show that no algorithm can have competitive ratio below~$2.18$ in general.
  
  In the second part of the paper, we focus on the inherently incremental greedy algorithm that increases the objective value as much as possible in each step.
  This algorithm is known to be $1.58$-competitive for submodular objective functions, but it has unbounded competitive ratio for the class of incremental problems mentioned above.
  We define a relaxed submodularity condition for the objective function, capturing problems like maximum (weighted) ($b$-)matching and a variant of the maximum flow problem.
  We show that the greedy algorithm has competitive ratio (exactly)~$2.313$ for the class of problems that satisfy this relaxed submodularity condition.
  
  Note that our upper bounds on the competitive ratios translate to approximation ratios for the underlying cardinality constrained problems.
\end{abstract}

\section{Introduction}
\global\long\def\OPT{\textsc{Opt}}
Practical solutions to optimization problems are often inherently
\emph{incremental} in the sense that they evolve historically instead
of being established in a one-shot fashion. This is especially true
when solutions are expensive and need time and repeated investments
to be implemented, for example when optimizing the layout of logistics
and other infrastructures. In this paper, we propose a theoretical
framework to capture \emph{incremental maximization} problems in some
generality. 

We describe an incremental problem by a set $\solelements$ containing
the possible elements of a solution, and an objective function $f\colon2^{\solelements}\to\mathbb{R}^{+}$
that assigns to each solution $S\subseteq\solelements$ some non-negative
value $f(S)$. We consider problems of the form

\begin{align}
\max &\; f(S)     \label{eq:problem} \\
\text{s.t.} &\; \left|S\right|\leq k \nonumber\\
 &\; S\subseteq\solelements, \nonumber
\end{align}

where $k\in\mathbb{N}$ grows over time.
 
An incremental solution $\vec{S}$ is given by an order $\{s_{1},s_{2},\dots\}:=\solelements$
in which the elements of $\solelements$ are to be added to the solution
over time. A good incremental solution needs to provide a good solution
after $k$ steps, for every $k$, compared to an optimum solution
$\opt{k}$ with $k$ elements, where we let $\opt{k}\in\arg\max_{S\subseteq\solelements,\left|S\right|=k}f(S)$ and $\vopt{k} := f(\opt{k})$.
Formally, we measure the quality of an incremental solution by its
\emph{competitive ratio.} 
For $\vec{S}_{k}:=\{s_{1},\dots,s_{k}\}\subseteq\solelements$ being the first~$k$ elements of~$\vec{S}$, 
we say that $\vec{S}$ is (strictly) $\rho$-competitive if
\[
\max_{k\in\{1,\dots,\left|\solelements\right|\}}\frac{\vopt{k}}{f(\vec{S}_{k})}\leq\rho.
\]
An algorithm is called $\rho$-competitive if it always produces a
$\rho$-competitive solution, and its \emph{competitive ratio} is
the infimum over all $\rho\geq1$ such that it is $\rho$-competitive.
Notice that we do not require the algorithm to run in polynomial time.

While all cardinality constrained optimization problems can be viewed
in an incremental setting, clearly not all such problems admit good
incremental solutions. For example, consider a cardinality constrained
formulation of the classical maximum $s$-$t$-flow problem: For a
given graph $G=(V,E)$, two vertices $s,t\in V$ and capacities $u\colon E\to\mathbb{R}^{+}$,
we ask for a subset $E'\subseteq E$ of cardinality $k\in\mathbb{N}$
such that the maximum flow in the subgraph $(V,E')$ is maximized.
The example in Figure~\ref{fig:incremental_flow} shows that we cannot
hope for an incremental solution that is simultaneously close to optimal
for cardinalities~1 and~2.
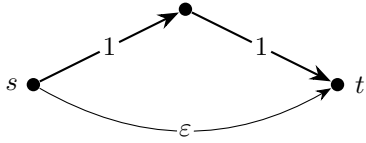
\begin{figure}
\begin{centering}
\begin{tikzpicture}[->]
\tikzstyle{every node} = [circle, fill = black, minimum size = 5, inner sep = 0]
\tikzset{above/.style = { label = {[label distance = 2]90:#1} } }
\tikzset{below/.style = { label = {[label distance = 2]270:#1} } }
\tikzset{>={Stealth[scale=1.2]}}
\node[label = {[label distance = 2]180:$s$}] (s) at (0, 0) {};
\node[label = {[label distance = 2]0:$t$}] (t) at (4, 0) {};
\node (v) at (2, 1) {};
\draw[thick] (s) edge node[midway, fill=white] {$1$} (v);
\draw[thick] (v) edge node[midway, fill=white] {$1$} (t);
\draw (s) edge[bend right] node[midway, fill=white] {$\varepsilon$} (t);
\end{tikzpicture}
\par\end{centering}
\caption{Example showing that the $s$-$t$-flow problem does not always admit
good incremental solutions, where $\varepsilon>0$ is arbitrarily
small. \label{fig:incremental_flow}}
\end{figure}

In order to derive general bounds on the competitive ratio of incremental problems, we need to restrict the class of objective functions~$f$ that we consider.
Intuitively, the unbounded competitive ratio in the flow example comes from the fact that we have to invest in the $s$-$t$-path of capacity~$1$ as soon as possible, but this path only yields its payoff once it is completed after two steps.

In order to prevent this and similar behaviors, we require $f$ to be monotone (i.e., $f(S)\leq f(T)$ if $S\subseteq T$) and sub-additive (i.e., $f(S)+f(T)\geq f(S\cup T)$).
Many important optimization problems satisfy these weak conditions, and we give a short list of examples below.
We will see that all these (and many more) problems admit incremental solutions with a bounded competitive ratio.
More specifically, we develop a general $2.618$-competitive incremental algorithm that can be applied to a broad class of problems, including all problems mentioned below.
We illustrate in detail how to apply our model to obtain an incremental variant of the matching problem, and then list incremental versions of other important problems that are obtained analogously.

\begin{itemize}
\item 
\underline{\noun{Maximum Weighted Matching}}: 
Consider a graph $G = (V,E)$ with edge weights $w\colon E \to \mathbb{R}_{\ge 0}$. 
If we think of edges as \emph{potential} connections and edge weights as \emph{potential} payoffs, then it is not enough to find the final matching because we cannot construct the edges all at once: the goal is to find a sequence of edges that achieves a high pay-off in the short, the medium, and the long term. In terms of our formal framework, we add edges to a set $S$ one at a time with $\solelements=E$ and $f(S)$ is the maximum weight of a matching $M \subseteq S$.
In order to be $\rho$-competitive, we need that, after $k$ steps for every~$k$, our solution~$S$ of cardinality~$k$ is no worse than a factor of~$\rho$ away from the optimum solution of cardinality $k$, i.e., $f(S) \geq f(S^{\star}_k)/\rho$.

This model captures the setting where the infrastructure (e.g. the matching, the knapsack, the covering, or the flow) must be built up over time.
The online model would be too restrictive in this setting because here we know our options in advance. 
Note that, as we add more edges, the set of edges~$S$ only needs to contain a large matching~$M$, but does not have to be a matching itself;
The matching~$M$ can change to an arbitrary subset of~$S$ from one cardinality to the next and does not have to stay consistent.
This ensures that $f(S)$ is monotonically increasing, and is in keeping with the infrastructures setting where the potential regret present in the online model does not apply: building more infrastructure can only help, since once it is built, we can change how it is used. 
Accordingly, in all the problems below the set $S$ does not have to be a valid solution to the cardinality constrained problem at hand, but rather needs to \emph{contain} a good solution as a subset. 
The objective~$f(S)$ is consistently defined to be the value of the best solution that is a subset of~$S$.
Notice that this approach can easily be generalized to \noun{Maximum $b$-Matching}.

\item
\underline{\noun{Set Packing}}: 
Given a set of weighted sets $\mathcal{X}$ we ask for an incremental subset $S\subseteq\mathcal{X}$ where $f(S)$ is the maximum weight of mutually disjoint subsets in $S$.
This problem captures many well-known problems such as \noun{Maximum Hypergraph Matching} and \noun{Maximum  Independent Set}.

\item
\underline{\noun{Maximum Coverage}}: 
Given a set of weighted sets $\mathcal{X}\subseteq 2^{U}$ over an universe of elements $U$, we ask for an incremental subset $S\subseteq\mathcal{X}$,
where $f(S)$ is the weight of elements in~$\bigcup_{X \in S} X$.
This problem  captures maximization versions of clustering and location problems.
We can include opening costs $c \colon \mathcal{X} \to \mathbb R_{\ge 0}$ by letting~$f(S)$ be the maximum over all subsets~$S'\subseteq S$ of the number (or weight) of the sets in~$S'$ minus their opening costs.

\item
\underline{\noun{Knapsack}}: 
Given a set $X$ of items, associated sizes $s\colon X\to\mathbb{R}_{\ge 0}$ and values $v\colon X\to\mathbb{R}_{\ge 0}$, and a knapsack of capacity $1$, 
we ask for an incremental subset $S\subseteq X$, where $f(S)$ is the largest value $\sum_{x\in S'}v(x)$ of any subset $S'\subseteq S$ with $\sum_{x\in S'}s(x)\leq1$.
This problem can be generalized to \noun{Multi-Dimensional Knapsack} by letting item sizes be vectors and letting the knapsack have a capacity in every dimension.

\item
\underline{\noun{Disjoint Paths}}:
Given a graph $G = (V,E)$, a set of pairs $\mathcal{X} \subseteq V^2$ with weights~$w\colon \mathcal{X} \to \mathbb{R}_{\ge 0}$, we ask for an incremental subset $S \subseteq \mathcal{X}$, where $f(S)$ is the maximum weight of a subset $S'\subseteq S$, such that $G$ contains mutually disjoint paths between every pair in~$S'$. 

\item
\underline{\noun{Maximum Bridge-Flow}}: 
We argued above that the maximum $s$-$t$-flow problem is not 
amenable
 to the incremental setting because it does not pay off to build paths partially. 

To overcome this, we consider a natural restriction of the flow problem where most edges are freely available to be used, and only the edges of a directed $s$-$t$-cut need to be built incrementally.
If the directed cut has no backward edges, every $s$-$t$-path contains exactly one edge that needs to be built, and we never have to invest multiple steps to establish a single path.
This problem captures logistical problems where links need to be established between two clusters, like when bridges need to be built across a river, cables across an ocean, or when warehouses need to be opened in a supplier-warehouse-consumer network.
Formally, given a directed graph $G=(V,E)$ with capacities $u\colon E\to\mathbb{R}$, vertices $s,t\in V$, and a directed $s$-$t$-cut $C\subseteq E$ induced by the partition $(U,W)$ of $V$ such that the directed cut induced by $(W,U)$ is empty, 
we ask for an incremental subset $S\subseteq C$ where $f(S)$ is the value of a maximum flow in the subgraph $(V,E\setminus(C\setminus S))$.

\end{itemize}

It is easy to verify that all the problems mentioned above (and many more) indeed have a monotone
and sub-additive objective function.
In addition, each one of these problems satisfies the following property: 
For every $S\subseteq\solelements$, there exists $s\in S$ with $f(S\setminus\{s\})\geq f(S)-f(S)/\left|S\right|$. We call this property the \emph{accountability} property -- to our knowledge, it has not been named before.
Intuitively, this property ensures that the value of a set $S\subseteq\solelements$ is the sum of individual contributions of its elements, and there cannot be additional value that emerges only when certain elements of $\solelements$ combine. 
While it is easy to formulate artificial problems that have monotonicity and sub-additivity but no accountability, 
we were not able to identify any natural problems of this kind. 
This justifies to add accountability to the list of properties that we require of incremental problems. 

\begin{definition} 
Given a set of elements $\solelements$, and a function $f \colon 2^{\solelements} \to \mathbb{R}$,
we say that the function $f$ is \emph{incremental} if it satisfies the following properties for every~$S,T \subseteq \solelements$:
\begin{enumerate}
\item (monotonicity): $S\subseteq T\Rightarrow f(S)\leq f(T)$,
\item (sub-additivity): $f(S)+f(T)\geq f(S\cup T)$,
\item (accountability): $\exists s\in S\colon f(S\setminus\{s\})\geq f(S)-f(S)/\left|S\right|$.
\end{enumerate}
We say that a cardinality constrained problem with increasing cardinality (eq.~\eqref{eq:problem}) is \emph{incremental} if its objective function is incremental.
\end{definition}

Observe that a $\rho$-competitive incremental algorithm immediately yields a $\rho$-approximation algorithm for the underlying cardinality constrained problem, with the caveat that the resulting approximation algorithm might not be efficient since we make no demands on the runtime of the incremental algorithm.
The converse is rarely the case since approximation algorithms usually do not construct their solution in incremental fashion.
A prominent exception are \emph{greedy} algorithms that are inherently incremental in the sense that they pick elements one-by-one such that each pick increases the objective by the maximum amount possible.
This type of a greedy algorithm has been studied as an approximation algorithm for many cardinality constrained problems, and approximation ratios translate immediately to competitive ratios for the incremental version of the corresponding problem.
In particular, the greedy algorithm is known to have competitive ratio (exactly) $\frac{e}{e-1}\approx1.58$
if the objective function $f$ is monotone and submodular~\cite{NemhauserWolseyFisher/78}.
Note, however, that of all the incremental problems listed above, only \noun{Maximum Coverage} (without opening costs) has a submodular objective function.
It is also known that if we relax the submodularity requirement and allow~$f$ to be the minimum of two monotone (sub-)modular functions, the greedy algorithm can be arbitrarily bad~\cite{KrauseMcMahanGuestrinGupta/08}.
We provide a different relaxation of submodularity that captures \noun{Maximum (Weighted) ($b$-)Matching} and \noun{Maximum Bridge-Flow}, and where the greedy algorithm has a bounded competitive/approximation ratio.

\subparagraph*{Our Results.}

As our first result, we show that every incremental problem admits a bounded competitive ratio. 

\begin{theorem}
\label{thm:main-incremental}
Every incremental problem admits a $(1+\phi)$-competitive algorithm, where $\phi$ is the golden
ratio and $(1+\phi)\approx2.618$. 
No general deterministic algorithm for this class of problems has a competitive ratio of $2.18$ or better.
\end{theorem}

Again, note that we make no guarantees regarding the running time of our incremental algorithm.
In fact, our algorithm relies on the ability to compute the optimum of the underlying cardinality constrained problem for increasing cardinalities.
If we can provide an efficient approximation of this optimum, we get an efficient incremental algorithm in the following sense.

\begin{corollary}
    \label{cor:polynomial_algorithm}
If there is a polynomial time $\alpha$-approximation algorithm for a cardinality constrained problem with incremental objective function, then we can design a polynomial time $\alpha (1+\phi)$-competitive incremental algorithm.
\end{corollary}

We also analyze the approximation/competitive ratio of the greedy algorithm.
We observe that for many incremental problems like \noun{Knapsack}, \noun{Maximum Independent Set}, and \noun{Disjoint Paths}, the greedy algorithm has an unbounded competitive ratio. 
On the other hand, we define a relaxation of submodularity called the $\alpha$-\emph{\weaksub} property under which the greedy algorithm has a bounded competitive ratio.
In particular, this relaxation captures our cardinality constrained versions of \noun{Maximum (Weighted) ($b$-)Matching} and \noun{Maximum Bridge-Flow}, where the incremental set~$S$ need not be feasible but only contain a good feasible subset.
We get the following result, where the tight lower bound for $\alpha=2$ is obtained for \noun{Maximum Bridge-Flow}. Notice that for $\alpha=1$, we obtain the $\frac{e}{e-1} \approx 1.58$ bound that is known for submodular functions. For $\alpha=2$, the bound is $\frac{2 e^2}{e^2-1} \approx 2.313$.

\begin{theorem}
    For every cardinality constrained problem with an $\alpha$-\weaksub objective (defined below), the greedy algorithm has approximation/competitive ratio \mbox{$\alpha\frac{ e^\alpha}{e^{\alpha}-1}$}.
    This bound is tight for the greedy algorithm on problems with 2-\weaksub objectives, which includes \noun{Maximum (Weighted) ($b$-)Matching} and \noun{Maximum Bridge-Flow}.
    \label{thm:greedy}
\end{theorem}

We emphasize that the families of instances we construct to obtain the lower bounds in Theorems~\ref{thm:main-incremental} and \ref{thm:greedy} require the number of elements to tend to infinity, since it takes time for incremental solutions to sufficiently fall behind the optimum solution.

\subparagraph*{Related Work}
Most work on incremental settings has focused on cardinality constrained \emph{minimization} problems. 
A prominent exception is the robust matching problem, introduced by Hassin and Rubinstein~\cite{HassinRubinstein2002}.
This problem asks for a weighted matching $M$ with the property that, for every value $k$, the total weight of the $\min(k,|M|)$ heaviest edges of $M$ comes close to the weight of a maximum weight matching of cardinality $k$. 
Note that this differs from our definition of incremental matchings in that the robust matching problem demands
that the ``incremental'' solution consists of a matching, while we allow any edge set that contains a heavy matching as a subset.
Since their model is more strict, all of the following competitive ratios carry over to our setting.
Hassin and Rubinstein~\cite{HassinRubinstein2002} gave an improved, deterministic algorithm that achieves competitive ratio $\sqrt{2}\approx1.414$. They also give a tight example for the $\sqrt{2}$ ratio, which also works in our incremental setting.
Fujita et al.~\cite{FujitaKobayashiMakino2013} extended this result to matroid intersection, and Kakimura and Makino~\cite{KakimuraMakino2013} showed that every independence system allows for a $\sqrt{\mu}$-competitive solution, with $\mu$ being the extendibility of the system. 
Matuschke at al.~\cite{MatuschkeSkutellaSoto/14} describe a randomized algorithm for this problem that, under the assumption that the adversary does not know the outcome of the randomness, has competitive ratio $\ln(4)\approx1.386$.

A variant of the knapsack problem with a similar notion of robustness was proposed by Kakimura et al.~\cite{Kakimura2012}.
In this problem a knapsack solution needs to be computed, such that, for every $k$, the value of the~$k$ most valuable items in the knapsack compares well with the optimum solution using $k$ items, for every $k$. 
Kakimura et al.~\cite{Kakimura2012} restrict themselves to polynomial time algorithms and show that under this restriction a bounded competitive ratio is possible only if the rank quotient of the knapsack system is bounded. 
In contrast, our results show that if we do not restrict the running time and if we only require our solution to \emph{contain} a good packing with $k$ items for
every $k$, then we can be $(1+\phi)$-competitive using our generic algorithm, even for generalizations like \noun{Multi-Dimensional Knapsack}. 
If we restrict the running time and use the well-known PTAS for the knapsack problem \cite{IbarraKim/75,Lawler/79}, we still get a $(1+\phi)(1+\varepsilon)$-competitive algorithm. 
Megow and Mestre~\cite{MegowMestre/13} and Disser et al.~\cite{DisserKlimmMegowStiller/17} considered another variant of the knapsack problem that asks for an order in which to pack the items that works well for every knapsack capacity. Kobayashi and Takizawa~\cite{KobayashiTakazawa2016} study randomized strategies for cardinality robustness in the knapsack problem.

Hartline and Sharp~\cite{HartlineSharp/07} considered an incremental variant of the maximum flow problem where capacities increase over
time. 
This is in contrast to our framework where the cardinality of the solution increases.

Incremental solutions for cardinality constrained minimization problems have been studied extensively, in particular for clustering~\cite{CharikarChekuriFederMotwani/04,DasguptaLong/05}, $k$-median~\cite{ChrobakKenyonNogaYoung/07,Fotakis2006,MettuPlaxton2003}, minimium spanning tree~\cite{BlumChalasaniCoppersmithEtAl/94,GoemansKleinberg/98}, and facility location~\cite{Plaxton/06}. 
An important result in this domain is the incremental framework given by Lin et al.~\cite{LinEtAl2010}.
This general framework allows to devise algorithms for every incremental minimization problem for which a suitable augmentation subroutine
can be formulated. 
Lin et al.~\cite{LinEtAl2010} used their framework to match or improve many of the known specialized bounds for the problems
above and to derive new bounds for covering problems. 
In contrast to their result, our incremental framework allows for a general algorithm that works out-of-the-box for a broad class of incremental maximization
problems and yields a constant (relatively small) competitive ratio.

Abstractly, incremental problems can be seen as optimization problems under uncertainty. 
Various approaches to handling uncertain input data have been proposed, ranging from robust and stochastic optimization to streaming and exploration. 
On this level, incremental problems can be seen as a special case of online optimization problems, i.e., problems where the input data arrives over time (see~\cite{BorodinElYaniv/98,FiatWoeginger/98}).
Whereas online optimization in general assumes adversarial input, incremental problems restrict the freedom of the adversary to deciding when to stop, i.e., the adversary may choose the cardinality $k$ while all other data is fixed and known to the algorithm. 
Online problems with such a ``non-adaptive'' adversary have been studied in other contexts~\cite{DaniHayes2006,FaigleKernTuran1989,HalldorssonShachnai2010}.
Note that online problems demand irrevocable decisions in every time step -- a requirement that may be overly restrictive in many settings where solutions develop over a long time period.
In contrast, our incremental model only requires a growing solution ``infrastructure'' and allows the actual solution to change arbitrarily over time within this infrastructure.

\section{A competitive algorithm for incremental problems}\label{sec:inc_algo}

In this section, we show the second part of Theorem~\ref{thm:main-incremental}, i.e., we give an incremental algorithm that is $(1+\phi \approx 2.618)$-competitive for all incremental problems.
For convenience, we define the \emph{density} $\delta_S$ of a set~$S\subseteq \solelements$ via~$\delta_S := f(S)/\left|S\right|$, and we let $\dopt{k} := \delta_{\opt{k}}$ denote the optimum density for cardinality~$k$.
Our algorithm relies on the following two observations that follow from the accountability of the objective function.

\begin{lemma}
    In every incremental problem and for every cardinality~$k$, there is an ordering~$\vec{S}^{\star}_k := \{s^{\star}_1,s^{\star}_2,\dots,s^{\star}_k\} := \opt{k}$, such that~$\delta_{\{s^{\star}_1,\dots,s^{\star}_i\}} \geq \delta_{\{s^{\star}_1,\dots,s^{\star}_{i+1}\}}$ for all~$i\in\{1,\dots,k-1\}$.
    We say that~$\vec{S}^{\star}_k$ is a \emph{greedy order} of~$\opt{k}$.\label{lem:opt_greedy_order}
\end{lemma}

\begin{proof}
    By accountability of~$f$, there is an element~$s^{\star}_k \in \opt{k}$ for which
    \[ \delta_{\opt{k} \setminus \{s^{\star}_k\}} = \frac{f(\opt{k}\setminus\{s^{\star}_k\})}{k - 1} \geq \frac{f(\opt{k})}{k} = \delta_{\opt{k}}. \]
    We can repeat this argument for~$s^{\star}_{k-1} \in \opt{k}\setminus\{s^{\star}_k\}$, $s^{\star}_{k-2} \in \opt{k}\setminus\{s^{\star}_k, s^{\star}_{k-1}\}$, etc.~to obtain the desired ordering~$\vec{S}^{\star}_k$.
\end{proof}

\begin{lemma}
    In every incremental problem and for every~$1 \leq k' \leq k$ we have~$\dopt{k'} \geq \dopt{k}$.\label{lem:decreasing_density}
\end{lemma}

\begin{proof}
    Fix any cardinality~$k > 1$.
    By accountability of the objective function~$f$, there is an element~$s^{\star} \in \opt{k}$ with
    \[ \dopt{k} = \frac{f(\opt{k})}{k} \leq \frac{f(\opt{k}\setminus\{s^{\star}\})}{k-1} \leq \frac{f(\opt{k-1})}{k-1} = \dopt{k-1}. \]
    It follows that~$\dopt{k}$ is monotonically decreasing in~$k$.
\end{proof}

Now, we define~$k_0 := 1$ and~$k_i := \lceil(1+\phi)k_{i-1}\rceil$ for all positive integers~$i$.
Our algorithm operates in phases~$i\in\{0,1,\dots\}$.
In each phase~$i$, we add the elements of the optimum solution~$\opt{k_i}$ of cardinality~$k_i$ to our incremental solution in greedy order (Lemma~\ref{lem:opt_greedy_order}). 
Note that we allow the algorithm to add elements multiple times (without effect) in order to not complicate the analysis needlessly (of course we would only improve the algorithm by skipping over duplicates).
In the following, we denote by $t_i$ the number of steps (possibly without effect) until the end of phase~$i$, i.e., we let~$t_0 := k_0$ and~$t_i := t_{i-1} + k_i$.

\begin{lemma}
    For every phase~$i\in\{0,1,\dots\}$, we have~$t_i \leq \phi k_i$.\label{lem:phase_length}
\end{lemma}

\begin{proof}
    We use induction over~$i$, with the case~$i=0$ being trivial, since~$t_0 = k_0$.
    Now assume that $t_{i-1} \leq \phi k_{i-1}$ for some~$i\geq 1$.
    Using the property~$\frac{\phi}{\phi + 1} = \phi - 1$ of the golden ratio, we get
    \[ t_i = t_{i-1} + k_i \leq \phi k_{i-1} + k_{i} \leq \frac{\phi}{\phi + 1} k_i + k_i = \phi k_i.\qedhere\]
\end{proof}

Finally, we show the solution~$\vec{S}$ computed by our algorithm is~$(1+\phi)$-competitive.

\begin{theorem}
    For every cardinality~$k$, we have $\vinc{k} \geq \vopt{k} / (1+\phi)$.
\end{theorem}

\begin{proof}
    We use induction over~$k$.
    The claim is true for~$k=t_0=1$, since~$\inc{1} = \opt{1}$ by definition of the algorithm.
    For the inductive step, we prove that if the claim is true for~$k=t_{i-1}$, then it remains true for all~$k\in\{t_{i-1} + 1, \dots, t_i\}$.
    Recall that~$k_i = \lceil(1+\phi)k_{i-1}\rceil$.
    By Lemma~\ref{lem:phase_length}, we have
    \[ t_{i-1} \leq \phi k_{i-1} < k_i < t_{i-1} + k_i = t_i, \]
    and we can therefore distinguish the following cases.
    
    \boldall{Case 1: $t_{i-1} < k < k_i$.} Since $k > t_{i-1}$, our algorithm has already completed phase~$i-1$ and added all elements of $\opt{k_{i-1}}$, so we have~$\vinc{k} \geq \vopt{k_{i-1}}$.
    Because~$k$ is an integer and $k < k_i = \lceil(1+\phi)k_{i-1}\rceil$, we have that~$k < (1+\phi)k_{i-1}$.
    By Lemma~\ref{lem:decreasing_density}, we thus have
    \[ \vopt{k} = \dopt{k}\cdot k < \dopt{k_{i-1}}\cdot(1+\phi)k_{i-1} = (1+\phi)\vopt{k_{i-1}} \leq (1+\phi)\vinc{k}. \]
    
    \boldall{Case 2: $k_i \leq k \leq t_i$.} 
    At time~$k$, our algorithm has already completed the first $k - t_{i-1}$ elements of $\opt{k}$.
    Since the algorithm adds the elements of~$\opt{k_i}$ in greedy order, we have~$\vinc{k} \geq (k - t_{i-1})\dopt{k_i}$.
    On the other hand, since~$k \geq k_i$, by Lemma~\ref{lem:decreasing_density} we have~$\vopt{k} = k\cdot\dopt{k} \leq k\cdot\dopt{k_i}$.
    In order to complete the proof, it is thus sufficient to show that~$k \leq (1+\phi)(k-t_{i-1})$.
    To see this, let~$k = k_i + k'$ for some non-negative integer~$k'$.
    Because~$t_{i-1}$ is integral, Lemma~\ref{lem:phase_length} implies~$t_{i-1} \leq \lfloor \phi k_{i-1} \rfloor$.
    Since~$\phi$ is irrational and~$k_{i-1}$ is integral, $\phi k_{i-1}$ cannot be integral, thus
    \[ k - t_{i-1} = k' + k_i - t_{i-1} \geq k' + \lceil(1+\phi)k_{i-1}\rceil - \lfloor \phi k_{i-1} \rfloor = k' + k_{i-1} + 1. \]
    This completes the proof, since
    \[ (1+\phi)(k - t_{i-1}) \geq (1+\phi)(k' + k_{i-1} + 1) > k' + (1+\phi)k_{i-1} + 1 \geq k' + k_i = k. \qedhere\]
\end{proof}

Corollary~\ref{cor:polynomial_algorithm} follows if we replace~$\opt{k_i}$ by an $\alpha$-approximate solution for cardinality~$k_i$.

\section{Lower bound on the best-possible competitive ratio}\label{sec:lower_bound}

In this section, we show the second part of Theorem~\ref{thm:main-incremental}, i.e., we give a lower bound on the best-possible competitive ratio for the maximization of incremental problems. 
For this purpose, we define the \noun{Region Choosing} problem. 
In this problem, we are given $\noregions$ disjoint sets $R_1,\dots,R_\noregions$, called \emph{regions}, with region $R_i$ containing $i$ elements with a value of $\densityr{i}$ each. 
We say that~$\delta(i)$ is the \emph{density} of region~$R_i$.
The total value of all elements in the region $R_i$ is $\vallb{i} := i \cdot \densityr{i}$ for all $i \in \regionsin$.

The objective is to compute an incremental solution $S \subseteq \solelements := \bigcup_{i=1}^{\noregions}R_{i}$ such that the maximum value of the items from a single region in $S$ is large. 
Formally, the objective function is given by $f(S) := \max_{i\in\regionsin}\left|R_{i}\cap
S\right|\cdot \densityr{i}$.\footnote{The following and all other proofs omitted from the main part of the paper are deferred to the appendix.} 

\begin{restatable}{observation}{obsregionchoosing}
 \noun{Region Choosing} is an incremental problem.
\end{restatable}

For our lower bound, we set $\densityr{i} := i^{\beta-1}$ for some $\beta \in (0,1)$ that we will choose later. For this choice of $\beta$, we have $\densityr{i} < \densityr{j}$ and $\vallb{i} > \vallb{j}$ for $0 \leq j < i \leq \noregions$. 
Also, for $\noregions \to \infty$ we have $\lim_{i\to\infty} \vallb{i} = \infty$. 
We call instances of the \noun{Region Choosing} problem in this form \emph{$\beta$-decreasing}.
Observe that in every $\beta$-decreasing instance the optimum solution of cardinality $i \leq \noregions$ is to take all $i$ elements from region $R_i$. 
This solution has value $\vopt{i} = i^\beta$.

In order to impose a lower bound on the best-possible competitive ratio for $\beta$-decreasing instances, we need some insights into the structure of incremental solutions with an optimal competitive ratio. 
First, consider a solution that picks only $i' < i$ elements from region~$R_i$. 
In this case, we could have picked~$i'$ elements from region~$R_{i'}$ instead -- this would only improve the solution, since densities are decreasing. 
Secondly, if we take~$i$ elements from region~$R_i$, it is always beneficial to take them in an uninterrupted sequence before taking any elements from a region~$R_j$ with~$j>i$: Our objective depends only on the region with the most value, therefore it never helps to take elements from different regions in an alternating fashion.
This leads us the following observation.

\begin{observation}
 For every $\beta$-decreasing instance of \noun{Region Choosing} there is an incremental solution with optimal competitive ratio of the following structure: 
For $k_0 < k_1 < \dots < k_\maxpicks \in \N$ with $\maxpicks\in\N$, it takes $k_0$ elements from region $R_{k_0}$, followed by $k_1$ elements from $R_{k_1}$, and so on, until finally $k_\maxpicks$ elements from region $R_{k_\maxpicks}$ are chosen. 
\end{observation}

Thus, we can describe an algorithm for the region-choosing problem by an increasing sequence of region indices $k_0,\dots,k_\maxpicks$. 
Note that, in order to have a bounded competitive ratio if $\noregions \to \infty$, we must have $\maxpicks \to \infty$, since $\lim_{i\to\infty}\vallb{i} = \infty$. 
We are interested in a cardinality for which an incremental solution given by $k_0,\dots,k_\maxpicks$ has a bad competitive ratio. 
We define
\[
 \far_i := \frac{1}{k_i}\sum_{j=0}^i k_j \qquad \text{for all } i \in \set{0,\dots,\maxpicks}.
\]
Observe that $\far_i > 1$ for all $i\in\{1,\dots,\maxpicks\}$.
We know that the value of the optimum solution for cardinality~$\far_i k_i$ is $\vallb{\far_i k_i} = (\far_i k_i)^\beta$, whereas the incremental solution only achieves a value of $\vallb{k_i} = (k_i)^\beta$. 
This allows us to derive the following necessary condition on the $\far_i$-values of $\rho$-competitive solutions.

\begin{observation}
 If an incremental solution defined by a sequence $k_0,\dots,k_\maxpicks$ is \mbox{$\rho$-com}petitive for some $\rho \geq 1$, we must have
 \begin{equation}
  \rho \geq \frac{\vallb{\far_i k_i}}{\vallb{k_i}} = \left(\frac{\far_i k_i}{k_i}\right)^\beta = \far_i^\beta  \quad \eq \quad \far_i \leq \rho^{\frac{1}{\beta}} \qquad \text{for all } i \in \set{0,\dots,\maxpicks}.
  \label{eq:rho_condition}
 \end{equation} 
\end{observation}

We will exclude a certain range of values of $\rho$ by showing that we can find a $\beta \in (0,1)$ such that, for a sufficiently large number of regions $\noregions$, necessary condition~\eqref{eq:rho_condition} is violated.
We do this by showing that, for some $i^\star \in \mathbb{N}$ and some fixed~$\eps > 0$, we have~$\alpha_{i+1} - \alpha_{i} > \eps$ for all~$i\geq i^{\star}$, i.e., as $i$ goes to $\infty$, condition~\eqref{eq:rho_condition} must eventually be violated.
The following definition relates a value of~$\beta \in (0,1)$ to a lower bound on the competitive ratio~$\rho$ for~$\beta$-decreasing instances.

\begin{definition}
    \label{def:problematic}
 A pair $(\rho, \beta)$ with $\rho \geq 1$ and $\beta \in (0,1)$ is \emph{problematic} if there is $\eps>0$ such that for all $x \in (1,\rho^{1/\beta}]$ it holds that~$\function_{\rho,\beta}(x)<0$, where 
 \begin{equation*}
  \function_{\rho,\beta}(x) := (\rho^{\frac{1}{\beta}} + \eps - x)^\frac{1}{1-\beta} - \frac{x}{x-1+\eps}.
 \end{equation*}
\end{definition}

We show that problematic pairs indeed have the intended property.

\begin{restatable}{lemma}{lemrhobeta}
  If $(\rho, \beta)$ is a problematic pair, then $\rho$ is a strict lower bound on the competitive ratio of incremental solutions for~$\beta$-decreasing instances of \noun{Region Choosing}.
  \label{lem:rho_beta}
\end{restatable}

All that remains is to specify a problematic pair in order to obtain a lower bound via Lemma~\ref{lem:rho_beta}.
It is easy to verify that $(2.18, 0.86)$ is a problematic pair.
Note that the resulting bound of~$2.18$ can slightly be increased to larger values below $2.19$.

\begin{theorem}
 There is no $2.18$-competitive incremental \noun{Region Choosing} algorithm. 
\end{theorem}

\section{The greedy algorithm for a subclass of incremental problems}\label{sec:greedy}

In this section, we analyze the greedy algorithm that computes an incremental solution $\vec{S}$ with $\vec{S}_{k}=\vec{S}_{k-1}\cup \{s_{k}\}$, where $s_{k}\in\arg\max_{s\in\solelements\setminus\vec{S}_{k-1}}f(\vec{S}_{k-1}\cup \{s_{k}\})$ and $\vec{S}_{0}=\emptyset$.
This algorithm is well-known to have competitive ratio $\frac{e}{e-1}\approx1.58$ if the objective function $f$ is monotone and submodular~\cite{NemhauserWolseyFisher/78}.
Note that every monotone and submodular function is incremental.
On the other hand, in general, the greedy algorithm does not have a bounded competitive ratio for incremental problems.

\begin{restatable}{observation}{unboundedgreedy}
    The greedy algorithm has an unbounded competitive ratio for many incremental problems, e.g., \noun{Knapsack}, \noun{Weighted Independent Set}, and \mbox{\noun{Disjoint Paths}}.
\end{restatable}

We will now define a subclass of incremental problems where the competitive ratio of greedy can be bounded.
Observe that submodularity of a function~$f\colon 2^{\solelements} \to \mathbb{R}_{\ge 0}$ implies that, for every $S \neq T \subseteq \solelements$, there exists an element~$t \in T \setminus S$ with~$f(S\cup\{t\}) - f(S) \geq (f(S \cup T) - f(S))/\left|T \setminus S\right|$.
Accordingly, we can define the following relaxation of submodularity.

\begin{definition}
 We say that $f\colon 2^{\solelements} \to \mathbb{R}_{\ge 0}$ is $\alpha$-\emph{\weaksub} for an $\alpha > 0$, if for every $S, T \subseteq \solelements$ with $T \setminus S \neq \emptyset$ there exists an element~$t \in T \setminus S$ with
 \begin{equation}\label{eq:greedycomp}
   f(S\cup\{t\}) - f(S) \geq \frac{f(S \cup T) - \alpha f(S)}{\left|T\right|}.
 \end{equation}	
\end{definition}

Thus, if $f$ is $\alpha$-\weaksub, we can improve a greedy solution if its value is more than a factor of $\alpha$ away from the value of an optimal solution. This definition is meaningful in the sense that it induces an interesting subclass of incremental problems.

\begin{lemma}
 The objective functions of \noun{Maximum (Weighted) ($b$-)Matching} and \noun{Maximum Bridge-Flow} are 2-\weaksub, but not submodular.\label{lem:weakly_submodular_problems}
\end{lemma}

We now show the first part of Theorem~\ref{thm:greedy}.

\begin{theorem}
If the objective function of an incremental problem is $\alpha$-\weaksub, the greedy algorithm is
$\alpha \frac{e^\alpha}{e^\alpha - 1}$-competitive.
\end{theorem}

\begin{proof}
 Let $\inc{i}$ be our greedy incremental solution after $i$ elements have been added.
Let us focus on an arbitrary cardinality $k > \alpha$, and say that for this cardinality
we have $\vopt{k} = \alpha\vinc{k} + \beta$, for some $\beta > 0$. For cardinalities $k \leq \alpha$, note that $\vopt{1} = \vinc{1}$ and that $\vopt{k} \leq k \vopt{1} \leq \alpha \vopt{1} \leq \alpha \vinc{k}$ due to sub-additivity. We will
show that we must have $\beta \leq \alpha\vinc{k}/(e^\alpha-1)$, which proves the theorem for cardinalities $k > \alpha$.

First, let us define $\profit{k} := \vinc{k} - \vinc{k-1}$ to be the additional value obtained by adding
the $k$-th element to our greedy solution $\inc{k}$. We claim that 
\begin{equation}\label{eq:profit-recursive}
 p_{i} \geq \frac{\beta + \alpha\sum_{j=i+1}^{k} p_j}{k-\alpha} \quad \text{for any positive integer $i < k$, and }\quad p_k \geq \frac{\beta}{k-\alpha}\ .
\end{equation} 
To prove \eqref{eq:profit-recursive}, we apply \eqref{eq:greedycomp} for $S = \inc{i-1}$ and $T = \opt{k}$, which guarantees the existence of a $t \in \opt{k} \setminus \inc{i-1}$ such that
\[ 
 f(\inc{i-1} \cup \set{t}) - \vinc{i-1} \geq \frac{f(\inc{i-1}\cup \opt{k}) - \alpha\vinc{i-1}}{k} \geq \frac{\vopt{k} - \alpha\vinc{i-1}}{k} \geq \frac{\beta + \alpha\sum_{j=i}^{k}p_j}{k}.
\]
The last inequality holds since we assume $\vopt{k} = \alpha\vinc{k} + \beta$ and for any $i \leq k$ we know that $\vinc{k} - \vinc{i-1} = \sum_{j=i}^{k} p_j$.
Since we construct $\inc{i}$ greedily, we have $f(\inc{i}) \geq f(\inc{i-1} \cup \set{t})$ and thus $\profit{i} \geq (\beta + \alpha\sum_{i \leq j \leq k} p_j)/k$. Rearranging to isolate $\profit{i}$ we get the bounds in \eqref{eq:profit-recursive}. Next, we claim that
\begin{equation}
\label{eq:profit-close}
p_i \geq \frac{\beta}{k-\alpha} \cdot \left(\frac{k}{k-\alpha}\right)^{k-i} \quad \text{for all $i \leq k$}.
\end{equation}
We prove this by induction for decreasing values of~$i$. The induction base for $i=k$ follows directly from \eqref{eq:profit-recursive}. 
For the induction step, we assume that the formula holds for all $p_j$ with $j \in \{i+1,\dots,k\}$. 
By \eqref{eq:profit-recursive} and the inductive hypothesis, we have
\begin{multline}
\label{eq:profits-closed}
(k-\alpha)\cdot p_{i} \geq \beta + \sum_{j=i+1}^{k} \alpha p_j 
           \geq \beta + \alpha\sum_{j=i+1}^{k} \frac{\beta}{k-\alpha} \cdot \left(\frac{k}{k-\alpha}\right)^{k-j} \\
	  	   = \beta + \frac{\alpha\beta}{k-\alpha}\sum_{j=0}^{k-i-1} \left(\frac{k}{k-\alpha}\right)^{j}
		   = \beta + \frac{\alpha\beta}{k-\alpha} \cdot \frac{(\frac{k}{k-\alpha})^{k-i} - 1}{\frac{k}{k-\alpha}-1}
		   = \beta \cdot \left(\frac{k}{k-\alpha}\right)^{k-i},
\end{multline}
which shows that the formula also holds for $p_{i}$, and thus completes the induction step.

We are now ready to prove the theorem. 
Recall that we assumed $\vopt{k} = \alpha\vinc{k} + \beta$ and want to show that $\beta < \alpha\vinc{k}/(e^\alpha-1)$.
We have
\[
 \vinc{k} = \sum_{i=1}^{k} p_i \stackrel{\eqref{eq:profit-close}}{\geq} \frac{\beta}{k-\alpha} \sum_{i=1}^{k} \left(\frac{k}{k-\alpha}\right)^{k-i} = 
\frac{\beta}{k-\alpha} \sum_{i=0}^{k-1} \left(\frac{k}{k-\alpha}\right)^{i} = \frac{\beta}{k-\alpha} \frac{(\frac{k}{k-\alpha})^{k} - 1}{\frac{k}{k-\alpha}-1}
\]
which can be rearranged to
\[
 \beta \leq \frac{(k-\alpha)\left(\frac{k}{k-\alpha}-1\right)\vinc{k}}{(\frac{k}{k-\alpha})^{k} - 1} = \frac{\alpha\vinc{k}}{(\frac{k}{k-\alpha})^{k} - 1}  \leq \frac{\alpha\vinc{k}}{e^\alpha - 1}
\]
The last inequality holds since for any $x > 0$ we have $(1 + 1/x)^{x+1} \geq e$, and thus for~$x=k/\alpha-1$ it follows that
$ e \leq \left(1+\frac{1}{\frac{k}{\alpha}-1}\right)^{\frac{k}{\alpha}} 
    = \left(1+\frac{\alpha}{k-\alpha}\right)^{\frac{k}{\alpha}} 
		= \left(\frac{k}{k - \alpha}\right)^{\frac{k}{\alpha}} 
$ and  thus~$(\frac{k}{k-\alpha})^k \geq e^\alpha$.

.
\end{proof}

\subsection{Lower bound}\label{sub:greedy_lower}

We now show the second part of Theorem~\ref{thm:greedy}, i.e., we show a matching lower bound on the competitive ratio of the greedy algorithm.
We show this by constructing the following family of instances for the \IMF problem.
For $k \in \mathbb{N}$, we define a graph $G_k = (V_k,E_k)$ with designated nodes $s$ and $t$ by
\ama
 V_k :=&\ \set{s,t} \cup \setc{v^1_i,v^4_i}{i=1,\dots,2k} \cup \setc{v^2_i,v^3_i}{i=1,\dots,4k},\\
 E_k :=&\ E^1_k \cup E^\infty_k \cup \bigcup_{i=1}^{2k} E_{k,i} \cup \bigcup_{i=1}^{2k} E'_{k,i},\\		
 E^1_k      &:= \setc{(s,v^2_i),(v^3_{3k+i},t)}{i=1,\dots,k},\\
 E^\infty_k &:= \setc{(s,v^2_{3k+i}),(v^2_i,v^3_i),(v^2_{3k+i},v^3_{3k+i}),(v^3_{i},t)}{i=1,\dots,k},\\
 E_{k,i}    &:= \set{(s,v^1_i),(v^1_i,v^2_{k+i}),(v^2_{k+i},v^3_{k+i}),(v^3_{k+i},v^4_i),(v^4_i,t)} \quad \text{ for all } i=1,\dots,2k,\\
 E'_{k,i}    &:= \setc{(v^1_i,v^2_j),(v^3_{3k+j},v^4_{i})}{j=1,\dots,k} \quad \text{ for all } i=1,\dots,2k.\\
\ema
The edge capacities $u_k\colon E_k \to \mathbb{R}_{\geq 0}$ are given by $u_k(e)=(\frac{k}{k-1})^{2k+1-i}$ for $e \in E_{k,i}$, by $u_k(e)=\frac{1}{k}(\frac{k}{k-1})^{2k+1-i}$ for $e \in E'_{k,i}$, by $u_k(e)=1$ for $e\in E^1_k$, and by $u_k(e)=\infty$ for $e\in E^\infty_k$.

For every $G_k$, we choose a directed $s$-$t$-cut $C_k := \setc{(v^2_i,v^3_i)}{i=1,\dots,4k}$. 
\Wlog, we will assume in the following that we can resolve all ties in the \Greedy algorithm to our preference. This can be done formally by adding some very small offsets to the edge weights, but we omit this for clarity. 
Now consider how the \Greedy algorithm operates on graph $G_k$.

\begin{restatable}{lemma}{greedybridgeflow}
    In step~$j\in\{1,\dots,2k\}$, the greedy algorithm picks edge~$(v_{k+j}^2,v_{k+j}^3)$.\label{lem:greedy_bridgeflow}
\end{restatable}

With this, we are ready to show the following result, which, together with Lemma~\ref{lem:weakly_submodular_problems}, implies the second part of Theorem~\ref{thm:greedy}.

\begin{theorem}
    The greedy algorithm has competitive ratio at least~$\frac{2e^2}{e^2-1}\approx 2.313$ for \noun{Maximum Bridge-Flow}.
\end{theorem}

\begin{proof}
    By Lemma~\ref{lem:greedy_bridgeflow}, the greedy algorithm picks the edges $(v^2_{k+1},v^3_{k+1}), \dots, (v^2_{3k},v^3_{3k})$ in the first~$2k$ steps.
    Thus, after step~$2k$, greedy can send an $s$-$t$-flow of value
\ama
 \sum_{i=1}^{2k} \left(\frac{k}{k-1}\right)^{i} = \left(\frac{\left(\frac{k}{k-1}\right)^{2k+1}-1}{\frac{k}{k-1}-1}-1\right) = (k-1)\left(\frac{k}{k-1}\right)^{2k+1}-k\ .
\ema 

On the other hand, the solution of size~$2k$ consisting of the edges~$(v^2_{1},v^3_{1}), \dots, (v^2_{k},v^3_{k})$, and $(v^2_{3k+1},v^3_{3k+1}), \dots, (v^2_{4k},v^3_{4k})$ results in an (optimal) flow value of
\ama
 2k+2\sum_{i=1}^{2k} \left(\frac{k}{k-1}\right)^{i} = 2(k-1)\left(\frac{k}{k-1}\right)^{2k+1}.
\ema 
This corresponds to a competitive ratio of 
\ama
 \frac{2(k-1)\left(\frac{k}{k-1}\right)^{2k+1}}{(k-1)\left(\frac{k}{k-1}\right)^{2k+1}-k} = \frac{2\left(\frac{k}{k-1}\right)^{2k}}{\left(\frac{k}{k-1}\right)^{2k}-1} = \frac{2\left(\frac{k}{k-1}\right)^{2(k-1)+2}}{\left(\frac{k}{k-1}\right)^{2(k-1)+2}-1}\ .
\ema
Substituting~$x := k-1$ and using the identity~$\lim_{x \to \infty} (1+1/x)^x = e$, we get the lower bound on the competitive ratio of the greedy algorithm claimed in Theorem~\ref{thm:greedy} in the limit:
\[
\lim_{x \to \infty} \frac{2\left(\frac{x+1}{x}\right)^{2x}\left(\frac{x+1}{x}\right)^2}{2\left(\frac{x+1}{x}\right)^{2x}\left(\frac{x+1}{x}\right)^2-1}
= \lim_{x \to \infty} \frac{2e^2\left(\frac{x+1}{x}\right)^2}{2e^2\left(\frac{x+1}{x}\right)^2-1}
= \frac{2e^2}{e^2 - 1}\ .
\qedhere\]
\end{proof}

\section{Conclusion}

We have defined a formal framework that captures a large class of incremental problems and allows for incremental solutions with bounded competitive ratio.
We also defined a new and meaningful subclass consisting of problems with $\alpha$-\weaksub objective functions for which the greedy algorithm has a bounded competitive ratio.
Hopefully our results can inspire future work on incremental problems from a perspective of competitive analysis. 

Obvious extensions of our results would be to close the gap between our bounds of~$2.618$ and~$2.18$ for the best-possible competitive ratio of incremental algorithms. 
In particular, it would be interesting whether or not the bound of~$2.313$ for the greedy algorithm in the $2$-\weaksub setting can be beaten by some other incremental algorithm already in the general setting.
Also, the $\alpha$-\weaksub and strictly submodular settings may allow for better incremental algorithms than the greedy algorithm.
Another question is whether abandoning the accountability condition yields an interesting class of problems.
Finally, it may be possible to generalize our framework to problems with a continuously growing budget and a cost associated with each element, instead of a growing cardinality constraint.
 
\subparagraph*{Acknowledgements}
We wish to thank Andreas B\"artschi and Daniel Graf for initial discussions and the the general idea that lead to the algorithm of Section~\ref{sec:inc_algo}.
We are also grateful for the detailed comments provided by an anonymous reviewer.

\pagebreak

\bibliography{literature}

\iftrue
\newpage
\appendix

\section{Proofs of Section~\ref{sec:lower_bound}}

\obsregionchoosing*

\begin{proof}
   The objective function of \noun{Region Choosing} is monotone by definition.
   
   Let~$S,T$ be two solutions to an instance of the \noun{Region Choosing} problem, and consider the region of maximum value in the solution~$S \cup T$.
   Let $v_X$ denote the total value of the items from this region in solution~$X$.
   Then,
   \[ f(S\cup T) = v_{S \cup T} \leq v_S + v_T \leq f(S) + f(T), \]
   which proves sub-additivity.
   
   To show accountability, let $R_S$ be the region of maximum value in solution~$S$.
   If $S \setminus R_S = \emptyset$, then, for every~$s\in S$, we have
   \[ f(S\setminus\{s\}) = f(S) - f(S)/|S|.\]
   Otherwise, for every~$s \in S\setminus R_S$, we have
   \[ f(S\setminus\{s\}) = f(S) \geq f(S) - f(S)/|S|.\]
\end{proof}

\lemrhobeta*

\begin{proof}
We fix a problematic pair~$(\rho,\beta)$ and let~$\eps$ be as in Definition~\ref{def:problematic}.
Consider a \mbox{$\beta$-de}creasing instance of sufficiently large size~$\noregions$ and assume that there is a $\rho$-competitive incremental solution for this instance, given by the sequence~$k_0,k_1,\dots,k_\maxpicks$.
Consider a cardinality $k = \sum_{j=0}^i k_j$ for any~$i \in \{1,\dots,\noregions\}$ for which the incremental solution takes all elements from regions $R_{k_0}$, $\dots$, $R_{k_i}$. 
Assume that we do not take any additional elements for larger cardinalities.
We are interested in the first cardinality for which $\rho$-competitiveness would be violated, i.e., where $\vallb{k_i}$ is not enough to be $\rho$-competitive. 
This is the minimal cardinality $t_i$ with $\vopt{\keytime_i} = t_i^\beta > \rho k_i^\beta = \rho \vallb{k_i}$, i.e.,
\begin{equation*}
 \keytime_i := \min \left\{t_i \in \N\ \left|\ t_i > \rho^{\frac{1}{\beta}} k_i\right.\right\} \leq \rho^{\frac{1}{\beta}}k_i + 1.
\end{equation*}	
Then, for cardinality~$\keytime_i$, the incremental solution must have taken enough value from a later region to be $\rho$-competitive. 
Without loss of generality, we can assume this region to be~$R_{k_{i+1}}$, otherwise the incremental solution that skips region~$R_{k_{i+1}}$ is also $\rho$-competitive, and we can consider this solution instead.
It follows that the incremental solution must satisfy
\[
 \left(\keytime_i - \sum_{j=0}^i k_j\right) \densityr{k_{i+1}} \geq \frac{1}{\rho} \vallb{t_i} > k_i^\beta,\]
 which, by definition of~$\alpha_i$, implies 
\begin{equation}
\left(\rho^{\frac{1}{\beta}} k_i + 1 - \far_i k_i\right) \densityr{k_{i+1}} > k_i^\beta.
\label{eq:beta_decreasing_cond1}
\end{equation}
Defining~$\ratio_i := k_i / k_{i-1}$ for all~$i\in\{1,\dots,\noregions\}$ gives~$\densityr{k_{i+1}} = \densityr{\ratio_{i+1} k_i} = \densityr{k_i} \cdot \ratio_{i+1}^{\beta-1}$. 
With this, eq.~\eqref{eq:beta_decreasing_cond1} can be written as
\[
\vallb{k_i} = k_i^\beta < \left(\rho^{\frac{1}{\beta}} + \frac{1}{k_i} - \far_i\right) \ratio_{i+1}^{\beta-1} \densityr{k_i} k_i
	= \left(\rho^{\frac{1}{\beta}} + \frac{1}{k_i} - \far_i\right) \ratio_{i+1}^{\beta-1} \vallb{k_i}.\]
Dividing by~$\vallb{k_i}$ yields
\begin{equation} 
   \label{eq:beta_decreasing_cond2}
   \ratio_{i+1} < \left(\rho^{\frac{1}{\beta}} + \frac{1}{k_i} - \far_i\right)^{\frac{1}{1-\beta}}.
\end{equation}
Since the incremental solution is~$\rho$-competitive, we have~$\far_i \in (1,\rho^{\frac{1}{\beta}}]$ by condition~\eqref{eq:rho_condition}.
Because $(\rho, \beta)$ is a problematic pair, we have
\begin{equation}
   \function(\far_i) = (\rho^{\frac{1}{\beta}} +\eps - \far_i)^{\frac{1}{1-\beta}} - \frac{\far_i}{\far_i-1+\eps} < 0.
   \label{eq:beta_decreasing_cond3}
\end{equation}
Observe that since $\lim_{i\to\infty} k_i = \infty$, we can find an $i^\star$ such that $\frac{1}{k_i} < \eps$ for all~$i\geq i^\star$.
Thus, for~$i\geq i^\star$ eqs.~\eqref{eq:beta_decreasing_cond3} and \eqref{eq:beta_decreasing_cond2} imply
\[ \ratio_{i+1} < (\rho^{\frac{1}{\beta}} +\eps - \far_i)^{\frac{1}{1-\beta}}  < \frac{\far_i}{\far_i - 1 + \eps}.\]
For~$i\geq i^\star$, we therefore get
\begin{align*}
   \far_{i+1} - \far_{i} &= \frac{1}{k_{i+1}} \sum_{j=0}^{i+1} k_j - \far_i \\
                         &= 1 + \frac{\far_i}{q_{i+1}} - \far_i\\
                         &> \eps.
\end{align*}
But this implies that, if~$\noregions$ is sufficiently large, condition~\eqref{eq:rho_condition} eventually gets violated, which contradicts that the incremental solution has competitive ratio~$\rho$.
\end{proof}

\section{Proofs of Section~\ref{sec:greedy}}

\unboundedgreedy*

\begin{proof}
   We construct a knapsack instance with three types of items for small~$\varepsilon>0$ and any~$k\in\mathbb{N}$: One item of size and value both $1-\varepsilon$, $k$ items of size~$2\varepsilon$ and value~$1-2\varepsilon$, and $k$ items of size and value both~$\varepsilon^2$.
   Obviously, the greedy algorithm first takes the largest item and then continues with the smallest one, since each of them further increases the maximum value that can be packed in the knapsack of capacity~$1$.
   Consequently, the greedy solution has value below~$1$ for cardinality~$k$, while the optimum value approaches~$k$.

  We can reproduce the same behavior for the \noun{Weighted Independent Set} problem by choosing a star of degree~$k$ plus $k$ isolated vertices as our input graph, where the center of the star has weight~$1-\varepsilon$, the leaves of the star have weight~$1-2\varepsilon$, and the isolated vertices have weight~$\varepsilon^2$.
  
 Similarly, for \noun{Disjoint Paths}, we can choose a long path and many isolated edges as input.
 The endpoints of the path form a pair of weight~$1-\varepsilon$, each edge along the path is a pair of weight~$1-2\varepsilon$, and each isolated edge has weight~$\varepsilon^2$.
\end{proof}

We separately show the four parts of Lemma~\ref{lem:weakly_submodular_problems}.

\begin{lemma}
 The objective function of \noun{Maximum Weighted $b$-Matching} is 2-\weaksub.
\end{lemma}

\begin{proof}
   Let~$G=(V,E)$ be a graph, let~$w\colon E \to \mathbb{R}_{\geq 0}$ be edge weights, let~$b\colon V \to \mathbb{N}$ be vertex capacities, and let $f \colon 2^E \to \mathbb{N}$ such that~$f(S)$ denotes the maximum weight of a $b$-matching in the subgraph~$(V,S)$ of~$G$.
   Consider two edge sets $S,T\subseteq E$ with~$S\neq T$ and let~$M_S\subseteq S$ be a maximum weight $b$-matching in the graph~$(V,S)$, i.e., no vertex~$v\in V$ is incident to more than~$b(v)$ edges in~$M_S$ and~$w(M_S) = f(S)$, where~$w(X) := \sum_{e\in X} w(e)$.
   We let $d_M\colon V \to \mathbb{N}$ denote the vertex degrees of the subgraph~$(V,M_S)$, and define $l_M(v)$ to be the weight of the $b$-matching~$M_S$ that we lose if we need to reduce~$b(v)$ by one, i.e., for every~$v\in V$ we define
   \[ l_M(v) :=
   \begin{cases}
      0, & \text{if }d_M(v) < b(v), \\
      \min_{e=\{v,w\}\in M} w(e), & \text{if }d_M(v) = b(v).
   \end{cases} 
   \]
   Now take any edge~$e=\{u,v\}$ of the maximum weight $b$-matching~$M_{S\cup T}$ in the graph~$(V,S \cup T)$, and assume that we need to add~$e$ to~$M_S$ without violating vertex capacities, i.e., we may first need to remove edges from~$M_S$ to make room for~$e$.
   If~$e$ is already part of~$M_S$, we do not need to change the matching, and, in particular, the weight of the matching remains unchanged.
   Otherwise, by definition of~$l_M$ we can ensure that the change in weight of the $b$-matching is at least
   \[ w(e) - l_M(u) - l_M(v). \]
   If we sum this change over all edges in~$M_{S \cup T} \setminus M_S$, and let $b'(v)$ denote the degree of~$v$ in the subgraph~$(V, M_{S \cup T}\setminus M_S)$, we obtain
   \begin{align*} 
       & w(M_{S \cup T}\setminus M_S) - \sum_{v\in V} b'(v)l_M(v) \\
       & \geq w(M_{S \cup T}\setminus M_S) - 2w(M_S \setminus M_{S \cup T}) \\
       & = w(M_{S \cup T}) - w(M_{S \cup T} \cap M_S) - 2w(M_S \setminus M_{S \cup T}) \\
       & = w(M_{S \cup T}) - w(M_S) - w(M_S \setminus M_{S \cup T}) \\
       & \geq w(M_{S \cup T}) - 2w(M_S) \\
       & = f(S \cup T) - 2f(S).
   \end{align*}
   Since~$M_S$ is a maximum matching in~$(V,S)$, no edge in~$S$ can have a positive contribution to this sum.
   If this expression is still positive, there must be an edge~$e \in M_{S \cup T}\setminus S$ that increases the weight of~$M_S$ by at least 
   \[ (f(S \cup T) - 2f(S)) / |M_{S \cup T}\setminus S| \geq (f(S \cup T) - 2f(S)) / |T \setminus S|, \]
   and we get
   \[ f(S \cup \{e\}) - f(S) \geq (f(S \cup T) - 2f(S)) / |T \setminus S|,\]
   as claimed.
   Otherwise, the right-hand side is not positive, and the inequality is trivially satisfied by monotonicity of~$f$.
\end{proof}

\begin{lemma}
 The objective function of \noun{Maximum Matching} is not submodular.
\end{lemma}

\begin{proof}
   To see this, consider a path~$P=(V,E)$ of length tree with edges $e_1,e_2,e_3\in E$ in this order along the path.
   Recall that~$f(X)$ is the cardinality of a maximum matching in the subgraph~$(V,X)$ for every~$X\subseteq E$.
   With this, it is easy to check that for~$S:=\{e_1,e_2\}$ and~$T:=\{e_2,e_3\}$ we have
   \[ f(S) + f(T) = 1 + 1 < 2 + 1 = f(S \cup T) + f(S \cap T), \]
   thus~$f$ is not submodular.
\end{proof}

\begin{lemma}
 The objective function of \noun{Maximum Bridge Flow} is 2-\weaksub.
\end{lemma}

\begin{proof}
Recall that for any subset $S \subset C$, $f(S)$ is the 
value of the maximum flow using edges in $E \sm (C \sm S)$.
For any set $X \subseteq C$, let $G_X$ be the graph that contains all the edges of $X$, 
plus all the edges in $G$ that are not in $C$:
Let $f_S$ be some maximum flow in $G_S$, and let $\val(f_S)$ be its value.
By definition of the \noun{Maximum Bridge-Flow} objective function we have 
$\val(f_S) = \flow(S)$.

Now, $f_S$ can also be viewed as a flow in $G_{S \cup T}$. 
Let $G^S_{S \cup T}$ be the residual graph in $G_{S \cup T}$ formed by flow $f_S$. 
Let $f_r$ be the maximum flow in $G^S_{S \cup T}$ ($r$ for residual), and let $\val(f_r)$ be its value.

By the properties of residual graphs, we know that $\flow(S \cup T) = \flow(S) + \val(f_r)$. 
Rearranging we get that $\val(f_r) = \flow(S \cup T) - \flow(S)$. Now, by the property of flows,
the flow $f_r$ can be decomposed into source-sink paths and cycles, and we can define $f_r$ to be a max flow in  $G^S_{S \cup T}$
that contains only paths (no cycles), because such a max flow must always exist. 
There are two types of paths to consider in the decomposition of $f_r$: some use backwards edges in $S$,
while others only forward edges in $S \cup T$. Note that the total capacity of all backwards residual edges
in $G^S_{S \cup T}$ is $\val(f_S) = \flow(S)$. Thus, if we let $\fnbr$ be the sub-flow of $f_r$ that
uses no backwards edges (nb for no backwards edges), and we let $\val(\fnbr)$ be its value, then we have 
\begin{equation}
\label{eq:no-backwards-edges}
\val(\fnbr) \geq \val(f_r) - \flow(S) = \flow(S \cup T) - 2\flow(S). 
\end{equation}

Now, note that the flow $\fnbr$ is decomposed into paths which each cross the cut $C$ exactly once, because
by the problem definition the cut $C$ is directed one way, and $\fnbr$ contains no backwards edges with which to go back to the source side of the cut. 
Moreover, none of these paths use any edges in $S$; if such a path existed, then it would use no edges in $C \sm S$ 
(because it crosses the cut exactly once), so it could have been added to $f_s$ in $G_S$, 
which contradicts $f_S$ being the maximum flow in $G_S$. Thus, every flow-path in $\fnbr$ uses a single edge in $T \sm S$
and no other edges in $C$. Now, for any edge $e \in T$, let $\fnbr(e)$ be the flow on $e$ in $\fnbr$. 
We now argue that $\flow(S \cup \set{e}) - \flow(S) \geq \fnbr(e)$. To see this, 
let $G^S_{S \cup \set{e}}$ be the residual graph of 
$G_{S \cup \set{e}}$ defined by flow $f_S$. 
By the properties of residual graphs we have that $\flow(S \cup \set{e}) - \flow(S)$ is precisely the value of the maximum flow in $G^S_{S \cup \set{e}}$.
But note that all the flow-paths in $\fnbr$ that go through edge $e$ do not go through any other edge in $T$, and so they are also valid flow-paths in $G^S_{S \cup \set{e}}$. Thus, the value of the maximum flow in $G^{S}_{S \cup \set{e}}$ is at least $\fnbr(e)$, as desired.

Since every flow-path in $\fnbr$ only goes through a single edge in $T$ we have $\val(\fnbr) = \sum_{e \in T}\fnbr(e)$.
Also, since~$f_S$ is a maximum flow in~$G_S$, we have $\fnbr(e)=0$ for $e \in S$, and thus~$\val(\fnbr) = \sum_{e \in T\setminus S}\fnbr(e)$.
Thus, there is some edge $e \in T$ with $\fnbr(e) \geq \val(\fnbr) / |T\setminus S|$, so by the argument in the paragraph above,
$\flow(S \cup \set{e}) - \flow(S) \geq \val(\fnbr) / |T \setminus S|$. Eq.~\eqref{eq:no-backwards-edges} then completes the lemma. 
\end{proof}

\begin{lemma}
 The objective function of \noun{Maximum Bridge Flow} is not submodular.
\end{lemma}

\begin{proof}
   Consider the graph $G=(V,E)$ in Figure~\ref{fig:bridge_flow}, where all arcs have capacity~$1$, and the directed cut is~$C=\{e_1,e_2,e_3\}$.
   Recall that the objective function of the \noun{Bridge-Flow} problem is defined such that~$f(X)$ is the value of a maximum flow in the graph~$(V,(E\setminus C)\cup X)$ for~$X\subseteq C$.
   With this, it is easy to check that for~$S:=\{e_1,e_2\}$ and~$T:=\{e_2,e_3\}$ we have
   \[ f(S) + f(T) = 1 + 1 < 2 + 1 = f(S \cup T) + f(S \cap T), \]
   thus~$f$ is not submodular.
   \begin{figure}
   \begin{centering}
   \begin{tikzpicture}[->]
   \tikzstyle{every node} = [circle, fill = black, minimum size = 5, inner sep = 0]
   \tikzset{above/.style = { label = {[label distance = 2]90:#1} } }
   \tikzset{below/.style = { label = {[label distance = 2]270:#1} } }
   \tikzset{>={Stealth[scale=1.2]}}
   \node[label = {[label distance = 2]180:$s$}] (s) at (0, 0) {};
   \node[label = {[label distance = 2]0:$t$}] (t) at (4.5, 0) {};
   \node (v1) at (1.5, 0) {};
   \node (v2) at (3, 0) {};
   \draw[thick] (v1) edge node[midway, fill=white] {$e_2$} (v2);
   \draw[thick] (s) edge[bend left = 45] node[midway, fill=white] {$e_1$} (v2);
   \draw[thick] (v1) edge[bend right = 45] node[midway, fill=white] {$e_3$} (t);
   \draw[thick] (s) edge (v1);
   \draw[thick] (v2) edge (t);
   \end{tikzpicture}
   \par\end{centering}
   \caption{Example of a non-submodular \noun{Maximum Bridge-Flow} instance with unit capacities. \label{fig:bridge_flow}}
   \end{figure}
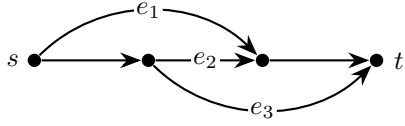
\end{proof}

\subsection{Proofs of Section~\ref{sub:greedy_lower}}

\greedybridgeflow*

\begin{proof}
   We prove the lemma by induction on~$j$, starting with step~$j=1$.
   Choosing $(v^2_{k+1},v^3_{k+1}) \in C_k$ in the first step results in a possible $s$-$t$-flow of $(\frac{k}{k-1})^{2k}$ along the $s$-$t$-path in $E_{k,1}$. 
   By construction, selecting $(v^2_{k+i},v^3_{k+i}) \in C_k$ with $i=2,\dots,2k$ yields less $s$-$t$-flow, since these edges have lower capacity. 
   Picking edge $(v^2_{k+i},v^3_{k+i}) \in C_k$ with $i=1,\dots,k,3k+1,\dots,4k$ results in a flow value bounded by the sum of the incoming edge capacities at~$v^2_{k+i}$ or the outgoing edge capacities at~$v^3_{k+i}$.
   In either case, we get a flow of
   \ama
   1+\frac{1}{k}\sum_{i=1}^{2k} \left(\frac{k}{k-1}\right)^{i} &= 1+\frac{1}{k} \left(\frac{\left(\frac{k}{k-1}\right)^{2k+1}-1}{\frac{k}{k-1}-1}-1\right)\\ 
   &= 1+\frac{1}{k} \left((k-1)\left(\frac{k}{k-1}\right)^{2k+1}-(k-1)-1\right) = \left(\frac{k}{k-1}\right)^{2k}.
   \ema
   Thus, no other edge results in more $s$-$t$-flow than edge~$(v^2_{k+1},v^3_{k+1})$, and with suitable tie-breaking we can ensure that the \Greedy algorithm picks this edge first, as claimed.
   
   Now assume that \Greedy picked the edges $(v^2_{k+1},v^3_{k+1}), \dots, (v^2_{k+j-1},v^3_{k+j-1})$ before step~$j$.
   Then, in step~$j$, it can pick edge~$(v^2_{k+j},v^3_{k+j})$ to increase the $s$-$t$-flow value by $(\frac{k}{k-1})^{2k+1-j}$ along the~$s$-$t$-path in~$E_{k,j}$.
   This is again better than selecting an edge $(v^2_{k+i},v^3_{k+i}) \in C_k$ for $i=j+1,\dots,2k$, since these edges have lower capacity. 
   For the remaining edges, we need to account for the fact that the edges $(s,v^1_i)$ and $(v^4_i,t)$ for $i=1,\dots,j-1$ are already saturated. 
   Therefore, the gain in flow value if we add any edge $(v^2_{k+i},v^3_{k+i}) \in C_k$ with $i=1,\dots,k,3k+1,\dots,4k$ in step~$j$ is
   \ama
   1+\frac{1}{k}\sum_{i=1}^{2k+1-j} \left(\frac{k}{k-1}\right)^{i} &= 1+\frac{1}{k} \left(\frac{\left(\frac{k}{k-1}\right)^{2k+2-j}-1}{\frac{k}{k-1}-1}-1\right)\\ 
   &= 1+\frac{1}{k} \left((k-1)\left(\frac{k}{k-1}\right)^{2k+2-j}-(k-1)-1\right) = \left(\frac{k}{k-1}\right)^{2k+1-j}.
   \ema
   This is again not better than picking edge~$(v^2_{k+j},v^3_{k+j})$, and with suitable tie-breaking we can ensure that the \Greedy algorithm picks~$(v^2_{k+j},v^3_{k+j})$, as claimed.
\end{proof}
\fi

\end{document}